\documentclass[12pt,fullpage]{article}
\usepackage{amsmath,amsfonts,amssymb,amsthm}
\usepackage[utf8]{inputenc}
\usepackage[osf,noBBpl]{mathpazo}
\usepackage{microtype}

\usepackage{xcolor}
\definecolor{darkgreen}{rgb}{0,0.6,0}

\usepackage[pdftex,bookmarks=true,bookmarksnumbered=true,pdfpagemode=UseOutlines,pdfstartview=FitH,
    pdfpagelayout=OneColumn,colorlinks=true,urlcolor=blue,citecolor=darkgreen,pdfborder={0 0 0},ocgcolorlinks]{hyperref}

\usepackage{complexity}

\theoremstyle{plain}
\newtheorem{theorem}{Theorem}
\newtheorem{lemma}[theorem]{Lemma}
\newtheorem{proposition}[theorem]{Proposition}
\newtheorem{corollary}[theorem]{Corollary}

\theoremstyle{definition}
\newtheorem{definition}[theorem]{Definition}
\newtheorem{remark}[theorem]{Remark}

\DeclareMathOperator{\val}{Val}
\DeclareMathOperator{\size}{size}
\newcommand\wronskian{\mathsf{W}}

\newcommand\FFp{\mathbb F_{\!p}}
\newcommand\FFps{\mathbb F_{\!p^s}}
\newcommand\KK{\mathbb K}
\newcommand\QQ{\mathbb Q}
\newcommand\ZZ{\mathbb Z}

\begin{document}

\title{Factoring bivariate lacunary polynomials without heights\thanks{Part of this work was done while the authors were visiting the University of Toronto.} }
\author{
    Arkadev Chattopadhyay\thanks{School of Technology and Computer Science, Tata Institute for Fundamental Research, \texttt{arkadev.c@tifr.res.in}.}
    \and Bruno Grenet\thanks{LIP, UMR 5668 ENS Lyon - CNRS - UCBL - INRIA, Université de Lyon, \texttt{\{bruno.grenet,pascal.koiran,natacha.portier\}@ens-lyon.fr}.}
    \and Pascal Koiran\footnotemark[3]
    \and Natacha Portier\footnotemark[3]
    \and Yann Strozecki\thanks{LRI -- Université Paris-Sud XI, \texttt{strozecki@logique.jussieu.fr}.}
}

\date{\today}
\maketitle

\begin{abstract}
 We present an algorithm which computes the multilinear factors of bivariate lacunary polynomials.
It is based on a new Gap theorem which allows to test whether $P(X)=\sum_{j=1}^k a_j X^{\alpha_j}(1+X)^{\beta_j}$
is identically zero in polynomial time.
The algorithm we obtain is more elementary than the one by Kaltofen and Koiran (ISSAC'05)
 since it relies on the valuation of polynomials of the previous form 
instead of the height of the coefficients. As a result, it can be used to find 
some linear factors of bivariate lacunary polynomials over a field of large finite characteristic
 in probabilistic polynomial time.
\end{abstract}

\newpage
\section{Introduction}

The \emph{lacunary}, or \emph{supersparse}, representation of a polynomial
\[P(X_1,\dotsc,X_n)=\sum_{j=1}^k a_j X_1^{\alpha_{1,j}}\dotsm X_n^{\alpha_{n,j}}\]
is the list of the tuples $(a_j, \alpha_{1,j},\dotsc,\alpha_{n,j})$ for $1\le j\le k$. This representation allows very high degree polynomials to be represented in a concise manner. The factorization of lacunary polynomials has been investigated in a series of papers. Cucker, Koiran and Smale first proved that integer roots of univariate integer lacunary polynomials can be found in polynomial time~\cite{CKS99}. This result was generalized by Lenstra who proved that low-degree factors of univariate lacunary polynomials over algebraic number fields can also be found in polynomial time~\cite{Len99}. More recently, Kaltofen and Koiran generalized Lenstra's results to bivariate and then multivariate lacunary polynomials~\cite{KK05,KK06}. A common point to these algorithms is that they all rely on a so-called \emph{Gap Theorem}: If $F$ is a factor of $P(\bar X)=\sum_{j=1}^t a_j \bar X^{\bar\alpha_j}$, then there exists $k_0$ such that $F$ is a factor of both $\sum_{j=1}^{k_0} a_j\bar X^{\bar\alpha_j}$ and $\sum_{j=k_0+1}^k a_j\bar X^{\bar\alpha_j}$. Moreover, the different Gap Theorems in these papers are all based on the notion of height of an algebraic number, and some of them use quite sophisticated results of number theory.

In this paper, we are interested in more elementary proofs for some of these results. We focus on Kaltofen and Koiran's first paper~\cite{KK05} dealing with linear factors of bivariate lacunary polynomials. We show how a Gap Theorem that does not depend on the height of an algebraic number can be proved. In particular, our Gap Theorem is valid for any field of characteristic zero. As a result, we get a new, more elementary algorithm for finding linear factors of bivariate lacunary polynomials over an algebraic number field. In particular, this new algorithm is easier to implement since there is no need to explicitly compute some constants from number theory, and the use of the Gap Theorem does not require to evaluate the heights of the coefficients of the polynomial. 
Moreover we use the same methods to prove a Gap Theorem for polynomials over some fields of positive characteristic, yielding an algorithm to find linear factors of bivariate lacunary polynomials 
of the form $(uX+vY+w)$ with $uvw\neq0$. 
Finding linear factors with $u=0$ is $\NP$-hard, and the same is true for 
linear factors with $v=0$ or $w=0$.
This follows from the fact that  finding univariate linear factors over finite fields is $\NP$-hard~\cite{KiSha99,BCR12,KL13}. 
In algebraic number fields we can find {\em all} linear factors in polynomial time, even those with $uvw=0$. For this we rely as Kaltofen and Koiran 
on Lenstra's univariate algorithm~\cite{Len99}.

Our Gap Theorem is based on the valuation of a univariate polynomial, defined as the maximum integer $v$ such that $X^v$ divides the polynomial. We give an upper bound on the valuation of a nonzero polynomial
\[P(X)=\sum_{j=1}^k a_j X^{\alpha_j}(1+X)^{\beta_j}.\]
This bound can be viewed as an extension of a result due to Haj\'os~\cite{Hajos,MS77}. We also note that Kayal and Saha recently used the valuation of square roots of polynomials to make some progress on the ``Sum of Square Roots'' problem~\cite{KS11}. 

Lacunary polynomials have also been studied with respect to other computational tasks. For instance, Plaisted showed the $\NP$-hardness of computing the greatest common divisor (GCD) of two univariate integer lacunary polynomials~\cite{Pla77}, and his results were extended to finite fields~\cite{vzGKS96,KaShp99,KK05}. On the other hand, some important special cases were identified for which the GCD of two lacunary polynomials can be computed in polynomial time~\cite{FGS08}. Other efficient algorithms for lacunary polynomials have been recently given, for instance for the detection of perfect powers~\cite{GR08,GR11} or interpolation~\cite{KN11}.

\paragraph{Acknowledgments.} We wish to thank Sébastien Tavenas for his help on Proposition~\ref{prop:LowerBound}, 
and Erich L. Kaltofen for pointing us out a mistake in Theorem~\ref{thm:FactorPosChar} in a previous version of this paper.

\section{Bound on the valuation}\label{sec:val}

In this section, we consider a field $\KK$ of characteristic zero and polynomials over $\KK$.

\begin{theorem}\label{thm:val}
Let $P=\sum_{j=1}^{k} a_j X^{\alpha_j} (1+X)^{\beta_j}$ with $\alpha_1\leq \dots \leq \alpha_k$.
If $P$ is not identically zero then its valuation is at most $\max_j (\alpha_j+\binom{k+1-j}{2})$.
\end{theorem}

A lower bound for the valuation of $P$ is clearly $\alpha_1$ (and it is attained when $\alpha_2>\alpha_1$ for instance). If the family $(X^{\alpha_j}(1+X)^{\beta_j})_{1\le j\le k}$ is linearly independent over $\KK$, the upper bound we get is actually $\alpha_1+\binom{k}{2}$: At most the first $\binom{k}{2}$ lowest-degree monomials can be cancelled. 
If $\alpha_j=\alpha_1$ for all $j$, Haj\'os' Lemma~\cite{Hajos,MS77} gives the better bound $\alpha_1+(k-1)$. (This bound can be shown to be tight by expanding ${X^{k-1}=(-1+(X+1))^{k-1}}$ with the binomial formula.) 
This is not true anymore when the $\alpha_j$'s are not all equal. One can show that the valuation can be as large as $\alpha_1+(2k-3)$ (see Proposition~\ref{prop:LowerBound}). 
The exact bound remains unknown, and whether this bound is still linear as in Haj\'os' Lemma or quadratic is open.

Our proof of Theorem~\ref{thm:val} is based on the so-called \emph{Wronskian} of a family of polynomials. This is a classical tool for the study of differential equations but it has recently been used to bound the valuation of a sum of square roots of polynomials~\cite{KS11} and also to bound the number of real roots of some sparse-like polynomials~\cite{KPT12}.

\begin{definition}
Let $f_1,\dotsc,f_k\in\KK[X]$. Their \emph{Wronskian} is the determinant of the \emph{Wronskian matrix}
\[\wronskian(f_1,\dotsc,f_k)=\det\begin{bmatrix}
    f_1         & f_2           & \dotsb & f_k     \\
    f_1'        & f_2'          & \dotsb & f_k'    \\
    \vdots      & \vdots        &        & \vdots  \\
    f_1^{(k-1)} & f_2^{(k-1)}   & \dotsb & f_k^{(k-1)} 
\end{bmatrix}.\]
\end{definition}

The main property of the Wronskian is its relation to linear independence. 
The following result is classical (see~\cite{BoDu10} for a simple proof of this fact).

\begin{proposition}\label{prop:wronskian}
The Wronskian of $f_1,\dotsc,f_k$ is nonzero if and only if the $f_j$'s are linearly independent over $\KK$.
\end{proposition}

The next two lemmas are our main ingredients to give a bound on the valuation for $P$, using a bound on the valuation of some Wronskian.

\begin{lemma}\label{lemma:valinf}
Let $f_1,\dotsc,f_k\in\KK[X]$. Then
\[\val(\wronskian(f_1,\dotsc,f_k))\ge \sum_{j=1}^k \val(f_j)-\binom{k}{2}.\]
\end{lemma}

\begin{proof}
Each term of the determinant is a product of $k$ terms, one from each column and one from each row. The valuation of such a term is at least $\sum_j\val(f_j)-\sum_{i=1}^{k-1} i$ since for all $i,j$, $\val(f_j^{(i)})\ge\val(f_j)-i$. The result follows.
\end{proof}

We can slightly refine the bound in this lemma. The term of valuation $\sum_j\val(f_j)-\binom{k}{2}$ in the Wronskian is indeed the determinant of the matrix made of the smallest degree monomials of each $f_j^{(i)}$. This determinant can vanish. In fact, one can easily see that this is the case if two $f_j$'s have the same valuation since this yields two proportional columns in the matrix. To use this idea more generally, consider that the $f_j$'s are ordered by increasing valuation. We define a \emph{plateau} to be a set $\{f_{j_0},\dots,f_{j_0+s}\}$ such that for $0<t\le s$, $\val(f_{j_0+t})\le\val(f_{j_0})+t-1$. The $f_j$'s are naturally partitioned into plateaux. Suppose that there are $(m+1)$ plateaux, of length $p_0$, \dots, $p_m$ respectively, and let $f_{j_0}$, \dots, $f_{j_m}$ their respective first elements. Generalizing the previous remark to plateaux, it can be shown that 
\begin{equation}\label{eq:refinement}
\val(\wronskian(f_1,\dotsc,f_k))\ge\sum_{i=0}^m \left(
p_i 
\val(f_{j_i})+\binom{p_i}{2}\right)-\binom{k}{2}.
\end{equation}
This bound is at least as large as in the lemma. If all the $f_j$'s have a different valuation, then the bound is equal to the bound stated in the lemma since there are in this case $k$ plateaux, each of length $1$. On the other side, if they all have the same valuation $\alpha$, there is one plateau of length $k$ and the bound is $\val(\wronskian(f_1,\dotsc,f_k))\ge k\alpha$. We investigate the implications of this refinement after the proof of Theorem~\ref{thm:val}.

\begin{lemma}\label{lemma:valsup}
Let $f_j=X^{\alpha_j}(1+X)^{\beta_j}$, $1\le j\le k$, such that $\alpha_j,\beta_j\ge k$ for all $j$. If the $f_j$'s are linearly independent, then
\[\val(\wronskian(f_1,\dotsc,f_k))\le \sum_{j=1}^k\alpha_j.\]
\end{lemma}

\begin{proof}
By Leibniz rule, for all $i,j$
\begin{equation}\label{eq:leibniz} 
f_j^{(i)}(X)=\sum_{t=0}^i \binom{i}{t} (\alpha_j)_t(\beta_j)_{i-t} X^{\alpha_j-t}(1+X)^{\beta_j-i+t}
\end{equation}
where $(m)_n=m(m-1)\dotsb(m-n+1)$ is the falling factorial.
Since $\alpha_j-t\ge\alpha_j-i$ and $\beta_j-i+t\ge\beta_j-i$ for all $t$, 
\begin{equation*}
f_j^{(i)}(X)=X^{\alpha_j-i}(1+X)^{\beta_j-i}
        \times \sum_{t=0}^i\binom{i}{t}(\alpha_j)_t(\beta_j)_{i-t} X^{i-1}(1+X)^t.
\end{equation*}
Furthermore, since $\alpha_j\ge k\ge i$, we can write $X^{\alpha_j-i}=X^{\alpha_j-k}X^{k-i}$ and since $\beta_j\ge k\ge i$, $(1+X)^{\beta_j-i}=(1+X)^{\beta_j-k}(1+X)^{k-i}$. Thus, $X^{\alpha_j-k}(1+X)^{\beta_j-k}$ is a common factor of the entries of the $j$-th column of the Wronskian matrix, and $X^{k-i}(1+X)^{k-i}$ is a common factor of the entries of the $i$-th row. Together, we get
\[\wronskian(f_1,\dotsc,f_k)=X^{\sum_j\alpha_j-\binom{k}{2}}(1+X)^{\sum_j\beta_j-\binom{k}{2}}\det(M)\]
where the matrix $M$ is defined by
\[M_{i,j}= \sum_{t=0}^i \binom{i}{t} (\alpha_j)_t(\beta_j)_{i-t} X^{i-t}(1+X)^t.\]
The polynomial $\det(M)$ is nonzero since the $f_j$'s are supposed linearly independent and its degree is at most $\binom{k}{2}$. Therefore its valuation cannot be larger than its degree and is bounded by $\binom{k}{2}$. 

Altogether, the valuation of the Wronskian is bounded by $\sum_j\alpha_j-\binom{k}{2}+\binom{k}{2}=\sum_j\alpha_j$.
\end{proof}

\begin{proof}[Proof of Theorem~\protect\ref{thm:val}]
Let $P=\sum_j a_jX^{\alpha_j}(1+X)^{\beta_j}$, and let $f_j=X^{\alpha_j}(1+X)^{\beta_j}$. 
We assume first that $\alpha_j,\beta_j\ge k$ for all $j$, and that the $f_j$'s are linearly independent. Note that $\val(f_j)=\alpha_j$ for all $j$.

Let $\wronskian$ denote the Wronskian of the $f_j$'s. We can replace $f_1$ by $P$ in the first column of the Wronskian matrix 
using column operations which multiply the determinant by $a_1$ (its valuation does not change). The matrix we obtain is the Wronskian matrix of $P,f_2,\dotsc,f_k$. Now using Lemma~\ref{lemma:valinf}, we get
\[\val(\wronskian)\ge \val(P)+\sum_{j\ge 2} \alpha_j-\binom{k}{2}.\]
This inequality combined with Lemma~\ref{lemma:valsup} shows that
\begin{equation}\label{eq:valLinIndep} 
\val(P)\le \alpha_1+\binom{k}{2}.
\end{equation}

We now aim to remove our two previous assumptions. If the $f_j$'s are not linearly independent, we can extract from this family a basis $f_{j_1}, \dots, f_{j_d}$. Then $P$ can be expressed in this basis as $P=\sum_{l=1}^d \tilde{a_l} f_{j_l}$. We can apply Equation~\eqref{eq:valLinIndep} 
to $f_{j_1}$,\dots, $f_{j_d}$ and obtain $\val(P)\le \alpha_{j_1} + \binom{d}{2}$. 
Since $j_d\le k$, we have $j_1+d-1\le k$ and $\val(P)\le\alpha_{j_1}+\binom{k+1-j_1}{2}$. The value of $j_1$ being unknown, we conclude that 
\begin{equation}\label{eq:valMax} 
\val(P)\le\max_{1\le j\le k}\left(\alpha_j+\binom{k+1-j}{2}\right).
\end{equation}

The second assumption is that $\alpha_j,\beta_j\ge k$. Given $P$, consider $\tilde P=X^k (1+X)^k P=\sum_j a_j X^{\tilde{\alpha_j}}(1+X)^{\tilde{\beta_j}}$. Then $\tilde P$ satisfies $\tilde{\alpha_j},\tilde{\beta_j}\ge k$, whence by Equation~\eqref{eq:valMax}, $\val(\tilde P)\le \max_j(\tilde{\alpha_j}+\binom{k+1-j}{2})$. Since $\val(\tilde P)=\val(P)+k$ and $\tilde{\alpha_j}=\alpha_j+k$, the result follows.
\end{proof}

\begin{remark}
In Theorem~\ref{thm:val}, we can replace $(1+X)$ by $(uX+v)$ for any $u,v\neq0$. Indeed, we can write $uX+v=v(\frac{u}{v} X+1)$ and then use the change of variables $Y=\frac{u}{v} X$. This gives us a polynomial of the same form as in the theorem, with the same valuation as the original one.
\end{remark}

\begin{remark}
Theorem~\ref{thm:val} does not hold in positive characteristic as shown by the equality $(1+X)^{2^n}+(1+X)^{2^{n+1}}=X^{2^n}(1+X) \mod 2$. Section~\ref{section:posChar} investigates the case of positive characteristic in more details.
\end{remark}

We argued after Lemma~\ref{lemma:valinf} that it can be refined. In the previous proof, it is used 
with $P$, $f_2$, \dots, $f_k$. If all the $f_j$'s have the same valuation $\alpha$, Equation~\eqref{eq:refinement} gives the bound $\val(\wronskian)\ge\val(P)+((k-1)\alpha+\binom{k-1}{2})-\binom{k}{2}$, whence $\val(P)\le\alpha+(k-1)$. In this case, replacing Lemma~\ref{lemma:valinf} by Equation~\eqref{eq:refinement} gives us a new proof of Haj\'os' Lemma, with the correct bound.

On the other hand, if the $f_j$'s have pairwise distinct valuations, Equation~\ref{eq:refinement} gives the same bound as Lemma~\ref{lemma:valinf}. Yet in this case 
Lemma~\ref{lemma:valsup} can be refined to obtain the bound $\val(\wronskian)\le\sum_j\alpha_j -\binom{k}{2}$. Again, we find the optimal bound for the valuation, that is $\val(P)=\alpha_1$ here. 

The refinement of Lemma~\ref{lemma:valinf} alone is not sufficient to improve Theorem~\ref{thm:val} in the general case. To this end, one needs to 
improve Lemma~\ref{lemma:valsup} as well. As already mentioned, it is an open problem to determine the best achievable bound for Theorem~\ref{thm:val}. 
The next proposition shows that it cannot be as low as in Haj\'os' Lemma.

\begin{proposition}\label{prop:LowerBound}
For $k\ge 3$, there exists a linearly independent family of polynomials $(X^{\alpha_j}(1+X)^{\beta_j})_{1\le j\le k}$, $\alpha_1\le\dotsb\le\alpha_k$ and a family of rational coefficients $(a_j)_{1\le j\le k}$ such that the polynomial 
\[P(X)=\sum_{j=1}^k a_j X^{\alpha_j} (1+X)^{\beta_j}\]
is nonzero and has valuation $\alpha_1+(2k-3)$.
\end{proposition}

\begin{proof}
A polynomial that achieves this bound is
\[P_k(X)=-1+(1+X)^{2k+3}-\sum_{j=0}^k a_j 
    X^{2j+1}(1+X)^{k+1-j},\]
where 
\[a_j=\frac{2k+3}{2j+1}\binom{k+1+j}{k+1-j}.\]
We aim to prove that $P_k(X)=X^{2k+3}$. Since it has $(k+3)$ terms and $\alpha_1=0$, this proves the proposition. To prove the result for an arbitrary value of $\alpha_1$, it is sufficient to multiply $P_k$ by some power of $X$.

It is clear that $P_k$ has degree $(2k+3)$ and is monic.
Let $[X^m]P_k$ be the coefficient of the monomial $X^m$ in $P_k$. Then for $m>0$
\[[X^m]P_k = \binom{2k+3}{m} - \sum_{j=0}^k 
                                    a_j \binom{k+1-j}{m-2j-1}.\]
We aim to prove that $[X^m]P_k=0$ as soon as $m<2k+3$. Using the definition of the $a_j$'s, this is equivalent to proving
\begin{equation}\label{eq:coeffPk}
\sum_{j=0}^k \frac{2k+3}{2j+1}\binom{k+1+j}{k+1-j}\binom{k+1-j}{m-2j-1} = \binom{2k+3}{m}.
\end{equation}

To prove this equality, we rely on Wilf and Zeilberger's algorithm~\cite{PetkovsekWilfZeilberger}, and its implementation in the Maple package \texttt{EKHAD} of Doron Zeilberger (see~\cite{PetkovsekWilfZeilberger} for more on this package). The program asserts the correctness of the equality and provides a recurrence relation satisfied by the summand that we can verify by hand. 

Let $F(m,j)$ be the summand in equation~\eqref{eq:coeffPk} 
divided by $\binom{2k+3}{m}$. We thus want to prove that $\sum_{j=0}^k F(m,j)=1$.
The \texttt{EKHAD} package provides 
\[R(m,j)=\frac{2j(2j+1)(k+j+2-m)}{(2k+3-m)(2j-m)}\]
and claims that
\begin{multline}\label{eq:WZ}
mF(m+1,j)-mF(m,j)\\=F(m,j+1)R(m,j+1)-F(m,j)R(m,j).
\end{multline}
In the rest of the proof, we show why this claim implies Equation~\eqref{eq:coeffPk}, and then that the claim holds.

Suppose first that Equation~\eqref{eq:WZ} holds and let us prove Equation~\eqref{eq:coeffPk}. If we sum Equation~\eqref{eq:WZ} for $j=0$ to $k$, we obtain
\begin{multline*}m(\sum_{j=0}^k F(m+1,j)-F(m,j))\\=F(m,k+1)R(m,k+1)-F(m,0)R(m,0).\end{multline*}
Since $R(m,0)=0$ and $F(m,k+1)=0$, $\sum_j F(m,j)$ is constant with respect to $m$. One can easily check that the sum is $1$ when $m=2k+2$. (Actually the only nonzero term in this case is for $j=k$.) Therefore, we deduce that for all $m<2k+3$,\footnote{The bound on $m$ is given by the fact that $R(m,j)$ is undefined for $m=2k+3$.} $\sum_j F(m,j)=1$, that is
Equation~\eqref{eq:coeffPk} is true.

To prove Equation~\eqref{eq:WZ}, note that
\[\frac{F(m+1,j)}{F(m,j)}= \frac{(j+k+2-m)(m+1)}{(m-2j)(2k+3-m)}\]
and 
\[\frac{F(m,j+1)}{F(m,j)}= \frac{(k+2-j)(m-2j-1)(m-2j-2)}{(2j+2)(2j+3)(j+k+3-m)}.\]
Therefore, to prove the equality, it is sufficient to check that
\begin{multline*}
0= m\frac{j+k+2-m}{m-2j} \frac{m+1}{2k+3-m}-m + R(m,j)\\
    - \frac{(k+2-j)(m-2j-1)(m-2j-2)}{(2j+2)(2j+3)(j+k+3-m)}R(m,j+1).
\end{multline*}
This is done by a mere computation.
\end{proof}

From Theorem~\ref{thm:val}, we can deduce the following Gap Theorem.

\begin{theorem}[Gap theorem] \label{thm:gap}
Let $P=\sum_{j=1}^ka_j X^{\alpha_j}(uX+v)^{\beta_j}$ 
with $u,v\neq 0$ and $\alpha_{j+1}\ge\alpha_j$, $0\le j<k$. Assume that there exists $\ell$ such that 
\begin{equation}\label{eq:gap} 
\alpha_{\ell+1}>\max_{1\le j\le\ell}\left(\alpha_j+\binom{\ell+1-j}{2}\right).
\end{equation}
Then $P$ is identically zero if and only if the polynomials $\sum_{j=1}^\ell a_j X^{\alpha_j}(uX+v)^{\beta_j}$ and $\sum_{j=\ell+1}^k a_j X^{\alpha_j}(uX+v)^{\beta_j}$ are both identically zero.

In particular, the smallest $\ell$ satisfying \eqref{eq:gap} is the smallest $\ell$ satisfying 
\[\alpha_{\ell+1}>\alpha_1+\binom{\ell}{2}.\]
\end{theorem}

\begin{proof}
Let $Q=\sum_{j=1}^\ell a_j X^{\alpha_j}(uX+v)^{\beta_j}$ and $R=P-Q$. 
Suppose that $Q$ is not identically zero. By Theorem~\ref{thm:val}, its valuation is at most $\max_j(\alpha_j+\binom{\ell+1-j}{2})$. 
Since $\alpha_j\ge\alpha_{\ell+1}$ for $j>\ell$, the valuation of $R$ is at least $\alpha_{\ell+1}>\max_j(\alpha_j+\binom{\ell+1-j}{2})$. 
Therefore, if $Q$ is not identically zero, its monomial of lowest degree cannot be canceled by a monomial of $R$. In other words, $P=Q+R$ is not identically zero.

For the second part of the theorem, consider the smallest $\ell$ satisfying Equation~\eqref{eq:gap}. It is clear that $\alpha_{\ell+1}>\alpha_1+\binom{\ell}{2}$. Moreover for all $j\le\ell$, $\alpha_{j+1}\le\max_{i\le j}(\alpha_i+\binom{j+1-i}{2})$. We now prove by induction on $j$ that $\alpha_j\le\alpha_1+\binom{j-1}{2}$ for all $j\le\ell$. This is obviously true for $j=1$. Let $j<\ell$ and suppose that for all $i\le j$, $\alpha_i\le\alpha_1+\binom{i-1}{2}$. Then 
\[\alpha_{j+1}\le\max_{i<j}\left(\alpha_i+\binom{j+1-i}{2}\right)\le\alpha_1+\max_{i<j}\left(\binom{i-1}{2}+\binom{j-(i-1)}{2}\right).\]
To conclude, we remark that $\binom{i-1}{2}+\binom{j-(i-1)}{2}\le\binom{j}{2}$ for all $i<j$.
\end{proof}

It is straightforward to extend this theorem to more \emph{gaps}. The theorem can be recursively applied to $Q$ and $R$ (as defined in the proof). Then, if $P=P_1+\dotsc+P_s$ where there is a \emph{gap} between $P_t$ and $P_{t+1}$ for $1\le t<s$, then $P$ is identically zero if and only if each $P_t$ is zero.

\section{Algorithms} \label{sec:algo}

In this section, we prove that there exists a deterministic polynomial-time algorithm to test if a polynomial of the form
\begin{equation} \label{eq:linear}
P=\sum_{j=1}^k a_j X^{\alpha_j} (uX+v)^{\beta_j},
\end{equation}
is identically zero and give a deterministic polynomial-time algorithm to compute the linear factors of a lacunary bivariate polynomial. The size of $P$ is defined by 
\begin{equation}\label{eq:sizeP} 
\size(P)=\size(u)+\size(v)+\sum_{j=1}^k(\size(a_j)+\log(\alpha_j\beta_j)).
\end{equation}
The algorithms use Lenstra's algorithm~\cite{Len99} or a variant of it for treating some special cases. This use of Lenstra's algorithm implies some restrictions on the field $\KK$ in which the coefficients of the polynomials lie. In this section, $\KK$ is an algebraic number field, and it is represented as $\KK=\QQ[\xi]/\langle\varphi\rangle$ where $\varphi\in\ZZ[\xi]$ is a monic irreducible polynomial. Elements of $\KK$ are given as vectors in the basis $(1,\xi,\dotsc,\xi^{\deg\varphi-1})$. That is for $e\in\KK$, $e=(e_0,\dotsc,e_{\deg\varphi-1})$ with $e_t=n_t/d_t$ for each $t$ where $n_t,d_t\in\ZZ$. Then 
\[\size(e) = \log(n_1d_1) + \dotsb + \log(n_{\deg\varphi-1}d_{\deg\varphi-1}).\]
The size of a polynomial defined as above is then approximately the number of bits needed to write down its binary representation.

Theorems~\ref{thm:pit} and \ref{thm:factorization} were already proven in~\cite{KK05}. We give here new proofs based on our Gap Theorem. The structures of the algorithms we propose are the same as in~\cite{KK05}. The only differences are the ones induced by the use of a different Gap Theorem. This implies some differences in terms of the complexity that are discussed at the end of this section.

\begin{theorem}\label{thm:pit}
There exists a deterministic polynomial-time algorithm to decide if a polynomial of the form~\eqref{eq:linear} is identically zero.
\end{theorem}

\begin{proof}
We assume without loss of generality that $\alpha_{j+1}\ge\alpha_j$ for all $j$ and $\alpha_1=0$. 
If $\alpha_1$ is nonzero, $X^{\alpha_1}$ is a factor of $P$ and we consider $P/X^{\alpha_1}$.

Suppose first that $u=0$. Then $P$ is given as a sum of monomials, and we only have to test each coefficient for zero. Note that the $\alpha_j$'s are not distinct. Thus the coefficients are of the form $\sum_j a_jv^{\beta_j}$. Lenstra~\cite{Len99} gives an algorithm to find low-degree factors of 
univariate lacunary 
polynomials. It is easy to deduce from his algorithm an algorithm to test such sums for zero. A strategy could be to simply apply Lenstra's algorithm to $\sum_ja_j X^{\beta_j}$ and then check whether $(X-v)$ is a factor, but one can actually improve the complexity by extracting from his algorithm the relevant part. 
The case $v=0$ is similar. 

We assume now that $u,v\neq 0$.
We split $P$ into small parts $P=P_1+\dotsb+P_s$, such that according to the Gap Theorem, $P$ is identically zero if and only if each part $P_t$ is identically zero. Formally, let $I_1$, \dots, $I_s$ be the (unique) partition of $\{1,\dotsc,k\}$ into intervals defined recursively as follows. Let $1\in I_1$. For $1\le j<k$, suppose that 
$\{1,\dotsc,j\}$ has been partitioned into $I_1,\dotsc,I_t$, and let $i_t$ be the smallest element of $I_t$. Then 
$(j+1)\in I_t$ if $\alpha_{j+1}\le\alpha_{i_t}+\binom{j-i_t+1}{2}$, 
and $(j+1)\in I_{t+1}$ otherwise. The polynomials $P_t=\sum_{j\in I_t} a_jX^{\alpha_j}(1+X)^{\beta_j}$ satisfy the conditions of Theorem~\ref{thm:gap}. 
Therefore, we are left with testing if the $P_t$'s are identically zero. Moreover, $X^{\alpha_{i_t}}$ divides $P_t$ for each $t$ and it is thus equivalent to be able to test if each $P_t/X^{\alpha_{i_t}}$ is identically zero.

To this end, let $Q$ be a polynomial of the form~\eqref{eq:linear} satisfying $\alpha_1=0$ and $\alpha_{j+1}\le\binom{j}{2}$ for all $j$. In particular, $\alpha_k\le\binom{k-1}{2}$.
Consider the change of variables $Y=uX+v$. Then
\[Q(Y)=\sum_{j=1}^k a_ju^{-\alpha_j} (Y-v)^{\alpha_j}Y^{\beta_j}\]
is identically zero if and only if $Q(X)$ is. 
We can express $Q(Y)$ as a sum of powers of $Y$: 
\[Q(Y)=\sum_{j=1}^k\sum_{\ell=0}^{\alpha_j} a_j u^{-\alpha_j} \binom{\alpha_j}{\ell} (-v)^\ell Y^{\alpha_j+\beta_j-l}.\]
There are at most $k\binom{k-1}{2}=\mathcal O(k^3)$ 
monomials. Then, testing if $Q(Y)$ is identically zero consists in testing each coefficient for zero. 
Moreover, each coefficient has the form $\sum_j \binom{\alpha_j}{\ell_j} a_j u^{-\alpha_j} (-v)^{\ell_j}$ where the sum ranges over at most $k$ indices. Since $\ell_j,\alpha_j\le\binom{k-1}{2}$ for all $j$, the terms in these sums have polynomial bit-lengths. 
Therefore, the coefficients can be tested for zero in polynomial time.

Altogether, this gives a polynomial-time algorithm to test 
if $P$ is identically zero.
\end{proof}

\begin{theorem}\label{thm:factorization}
Let 
\[P(X,Y)=\sum_{j=1}^k a_j X^{\alpha_j}Y^{\beta_j}\in\KK[X,Y].\]
There exists a deterministic polynomial-time algorithm that finds all the linear factors of $P$, together with their multiplicities.
\end{theorem}

\begin{proof}
A linear factor of $P$ is either of the form $(Y-uX-v)$ or $(X-a)$. To search factors of the form $(X-a)$, we see $P$ as a univariate polynomial in $Y$ whose coefficients are univariate polynomials in $X$. Then, $(X-a)$ is a factor of $P$ if and only if it is a factor of all the coefficients of $P$ viewed as a polynomial in $Y$. Lenstra gives an algorithm to compute linear factors of univariate lacunary polynomials~\cite{Len99}. Thus, we can find all the factors of the form $(X-a)$ and their multiplicities using his algorithm.

Now $(Y-uX-v)$ is a factor of $P$ if and only if $P(X,uX+v)$ vanishes identically. 
We can assume that $u\neq 0$. If $v=0$, $P(X,uX)=\sum_j a_ju^{\beta_j} X^{\alpha_j+\beta_j}$. Therefore, it vanishes if and only if each coefficient vanishes. But a coefficient of this polynomial is of the form $\sum_j a_ju^{\beta_j}$. Testing such a coefficient for zero is done in polynomial time using Lenstra's algorithm as in the proof of Theorem~\ref{thm:pit}, and there are at most $k$ of them to test.

Suppose now that $u,v\neq 0$.
Since $P(X,uX+v)$ is of the form~\eqref{eq:linear}, we can use our Gap Theorem (Theorem~\ref{thm:gap}) as in the proof of Theorem~\ref{thm:pit}: Let $P=\sum_{i=1}^s X^{\alpha_{(i)}} P_i$ where each $P_i$ is of the form~\eqref{eq:linear} and satisfies $\alpha_1=0$ and $\alpha_k\le\binom{k-1}{2}$. 
Then by Theorem~\ref{thm:gap}, $P(X,uX+v)$ vanishes if and only if $P_i(X,uX+v)$ vanishes for every $i$. Now apply the same transformation to each $P_i$, inverting the roles of $X$ and $Y$. Then each $P_i$ can be written as the sum $\sum_{\ell=1}^{s_i} Y^{\beta_{(\ell)}}P_{i\ell}$ where each $P_{i\ell}$ is of the form~\eqref{eq:linear} and satisfies $\alpha_1=\beta_1=0$ and $\alpha_k,\beta_k\le\binom{k-1}{2}$. 
Furthermore, $P(X,uX+v)$ vanishes if and only if all the $P_{i\ell}(X,uX+v)$ vanish.

Since the $P_{i\ell}$'s are low-degree polynomials, and there are at most $k$ of them, one can find all their linear factors. 
This relies on one of the numerous deterministic polynomial-time algorithms to factor dense multivariate polynomials that appear in the literature, from \cite{Kal85,Len87} to \cite{Gao03,Lec07}.
By the above discussion, the linear factors of $P$ are exactly the linear factors that all the $P_{i\ell}$'s have in common. Several strategies can be used to find these linear factors: Either we search the linear factors of all the $P_{i\ell}$'s and keep only the ones they have in common, or we search the linear factors of one particular $P_{i\ell}$ (for instance the one of smallest degree) and test if they are factors of the other $P_{i\ell}$'s using our PIT algorithm, or we compute the gcd of all the $P_{i\ell}$'s and then search its linear factors. In particular, this last solution directly gives the multiplicities of the factors of $P$, since it is the same as their multiplicities in the gcd.
\end{proof}

As Kaltofen and Koiran's algorithm~\cite{KK05}, our algorithm uses Lenstra's algorithm for univariate lacunary polynomials~\cite{Len99} to find univariate factors of the input polynomial. 
To compare both algorithms, let us thus focus on the task on finding truly bivariate linear factors, that is of the form $(Y-uX-v)$ with $uv\neq 0$.  

A first remark concerns the simplicity of the algorithm. 
The computation of the gap function is much simpler in our case since we do not have to compute the height of the coefficients. This means that the task of finding the gaps in the input polynomial is reduced to completely combinatorial considerations. 

Both our and Kaltofen and Koiran's algorithms use a dense factorization algorithm as a subroutine. This is in both cases the main computational task since the rest of the algorithm is devoted to the computation of the gaps in the input polynomial. Thus, a relevant measure to estimate the complexity of these algorithms is the maximum degree of the polynomials given as input to the dense factorization algorithm. This maximum degree is given by the values of the gaps in the two Gap Theorems. In our algorithm, the maximum degree is $\binom{k}{2}$. In Kaltofen and Koiran's, it is $\mathcal O(k\log k+k\log h_P)$ where $h_P$ is the \emph{height} of the polynomial $P$ and the value $\log(h_P)$ is a bound on the size of the coefficients of $P$. For instance, if the coefficients of $P$ are integers, then $h_P$ is the maximum of their absolute values.
Therefore, our algorithm has a better asymptotic complexity as soon as the size of the coefficients exceeds the number $k$ of terms. Furthermore, the hidden constant in the bound for Kaltofen and Koiran's algorithm is only known to be bounded by approximately $15$ while the corresponding constant in our case is $1/2$.

Note that an improvement of Theorem~\ref{thm:val} to a linear bound instead of a quadratic one would give us a better complexity than Kaltofen and Koiran's algorithm for all polynomials. Finally, it is naturally possible to combine both Gap Theorems in order to obtain the best complexity in all cases.

\section{Generalizations}

In this section, we aim to prove some generalizations of the results obtained in Sections~\ref{sec:val} and \ref{sec:algo}. 
The field $\KK$ is still supposed to be an algebraic number field as in Section~\ref{sec:algo}, 
unless otherwise stated.

Our first generalization shows that the identity test algorithm of Theorem~\ref{thm:pit} can be extended to a slightly more general family of polynomials. Namely, the linear polynomial $(uX+v)$ can be replaced by any $2$-sparse polynomial. 

\begin{theorem} 
Let $P=\sum_{j=1}^k a_j X^{\alpha_j} (uX^d+v)^{\beta_j}$.
There exists a deterministic polynomial-time algorithm to decide if 
the polynomial $P$
is identically zero.
\end{theorem} 

In the theorem, $(uX^d+v)$ could be replaced by the seemingly more general expression $(uX^d+vX^{d'})$ with $d>d'>0$. Yet, in this case we can factor out $X^{d'}$. A term $X^{\alpha_j}(uX^d+vX^{d'})^{\beta_j}$ can thus be written $X^{\alpha_j+d'\beta_j} (uX^{d-d'}+v)^{\beta_j}$. This has the same form as in the theorem, replacing $\alpha_j$ by $(\alpha_j+d'\beta_j)$ and $d$ by $(d-d')$.

The size of the polynomial in the statement of the theorem is defined as in Equation~\eqref{eq:sizeP} of 
Section~\ref{sec:algo} with the additional term $\log d$ in the sum. This means that the complexity of the algorithm is still polylogarithmic in the degree. 

\begin{proof}
For all $j$ we consider the Euclidean division of $\alpha_j$ by $d$: $\alpha_j=q_jd+r_j$ with $r_j<d$. We rewrite $P$ as
\[P=\sum_{j=1}^k a_j X^{r_j} (X^d)^{q_j} (uX^d+v)^{\beta_j}.\]
Let us group in the sum all the terms with a common $r_j$. That is, let 
\[P_i(Y)=\sum_{\substack{1\le j\le k\\ r_j=i}} a_j Y^{q_j} (uY+v)^{\beta_j}\]
for $0\le i<d$. We remark that regardless of the value of $d$, the number of nonzero $P_i$'s is bounded by $k$. We have $P(X)=\sum_{i=0}^{d-1} X^i P_i(X^d)$. Each monomial $X^\alpha$ of $X^iP_i(X^d)$ satisfies $\alpha\equiv i\mod d$. Therefore, $P$ is identically zero if and only if all the $P_i$'s are identically zero.

Since each $P_i$ is of the form~\eqref{eq:linear}, and there are at most $k$ of them, we can apply the algorithm of Theorem~\ref{thm:pit} to each of them.
\end{proof}

We now state a generalization of Theorem~\ref{thm:val}. A special case of this generalization is used in the following to extend our factorization algorithm of Theorem~\ref{thm:factorization}. 
It is not known whether the most general version of the theorem can be used to further extend our algorithms to be able to find small-degree factors of lacunary polynomials.

Note that this result holds whatever field $\KK$ of characteristic zero is considered.

\begin{theorem}\label{thm:generalization}
Let 
$(\alpha_{ij})\in\ZZ_+^{m\times k}$ and
\[P=\sum_{j=1}^k a_j \prod_{i=1}^m f_i^{\alpha_{ij}}\in\KK[X],\]
where the degree of $f_i\in\KK[X]$ is $d_i$ for all $i$. 
Let $\xi\in\KK$ and denote by $\mu_i$ the multiplicity of $\xi$ as a root of $f_i$. Then the multiplicity $\mu_P(\xi)$ of $\xi$ as a root of $P$ satisfies
\[\mu_P(\xi)\le \max_{1\le j\le k}\ \sum_{i=1}^m\left( \mu_i\alpha_{ij}+(d_i-\mu_i)\binom{k+1-j}{2}\right).\]
\end{theorem}

A proof of this theorem is given in Appendix~\ref{app:generalization}. Note that it can be stated in the more general settings of rational exponents $\alpha_{ij}$. It can then be seen as a generalization of a result of Kayal and Saha~\cite[Theorem~2.1]{KS11}.

The following corollary, used to find multilinear factors of bivariate lacunary polynomials, is a direct consequence of the theorem.

\begin{corollary} \label{lemma:prod3}
Let $P=\sum_{j=1}^k a_j X^{\alpha_j}(uX+v)^{\beta_j}(wX+t)^{\gamma_j}$, $uvwt\neq0$. If $P$ is nonzero, its valuation is at most
$\max_{1\le j\le k} (\alpha_j+2\binom{k+1-j}{2})$.
\end{corollary}

We now describe how to use this corollary to get a new factorization algorithm. Compared to Theorem~\ref{thm:factorization}, we are now able to find the multilinear factors instead of the linear ones.

\begin{theorem}\label{thm:multilinear}
Let $P=\sum_{j=1}^k a_j X^{\alpha_j}Y^{\beta_j}$. There exists a deterministic polynomial time algorithm to compute all the \emph{multilinear} factors of $P$, with multiplicity.
\end{theorem}

\begin{proof}[Proof sketch]
The proof goes along the same lines as the proof of Theorem~\ref{thm:factorization}. Suppose that $XY-(aX-bY+c)$ is a factor of $P$. Then the rational function $P(X,\frac{aX+c}{X+b})$ vanishes identically. Let us assume for simplicity that $a,b,c\neq 0$. (The other cases can be handled separately, as in the proof of Theorem~\ref{thm:factorization}.) Let 
\[Q(X)=(X+b)^{\max_i\beta_i} P(X,\frac{aX+c}{X+b})=\sum_{j=1}^k a_j X^{\alpha_j}(aX+c)^{\beta_j}(X+b)^{\gamma_j}\] 
where $\gamma_j=\max_i(\beta_i) -\beta_j$. Then $Q$ is a polynomial and it vanishes if and only if the rational function $P(X,\frac{aX+c}{X+b})$ does. By Corollary~\ref{lemma:prod3}, if $Q$ is nonzero its valuation is at most $\max_j (\alpha_j+2\binom{k+1-j}{2})$. We can deduce a Gap Theorem: For $1\le k_0\le k$, let
\[Q_0(X)=\sum_{j=1}^{k_0} a_j X^{\alpha_j}(aX+c)^{\beta_j}(X+b)^{\gamma_j}\]
and $Q_1=Q-Q_0$. Suppose that $\alpha_{k_0+1} > \max_{1\le j\le k_0}(\alpha_j+2\binom{k_0+1-j}{2})$. Then $Q$ vanishes identically if and only if $Q_0$ and $Q_1$ both vanish identically. Hence, $XY-(aX-bY+c)$ is a factor of $P$ if and only if it is a factor of both $P_0$ and $P_1$, defined by analogy with $Q_0$ and $Q_1$: $P_0$ is the sum of the $k_0$ first terms of $P$ and $P_1$ the sum of the $(k-k_0)$ last terms.

This proves that $P$ can be written as a sum of $P_{i\ell}$'s as in the proof of Theorem~\ref{thm:factorization} such that the multilinear factors of $P$ are the common multilinear factors of the $P_{i\ell}$'s, and such that each $P_{i\ell}$ is of the same form as $P$ and satisfies $\alpha_k,\beta_k\le 2\binom{k-1}{2}$. It thus remains to find the common multilinear factors of some low-degree polynomials. Since there are at most $k$ of them, this can be done in polynomial time.
\end{proof}

\section{Positive characteristic}\label{section:posChar}

As mentioned earlier, Theorem~\ref{thm:val} does not hold in positive characteristic. We considered the polynomial $(1+X)^{2^n}+(1+X)^{2^{n+1}}=X^{2^n}(X+1)$ in characteristic $2$. It only has two terms, but its valuation equals $2^n$. Therefore, its valuation cannot be bounded by a function of the number of terms. Note that this can be generalized to any positive characteristic. In characteristic $p$, one can consider the polynomial $\sum_{i=1}^p (1+X)^{p^{n+i}}$. 

Nevertheless, the exponents used in all the examples depend on the characteristic. In particular, the characteristic is always smaller than the largest exponent that appears. We shall show that in large characteristic, Theorem~\ref{thm:val} still holds. This contrasts with the previous result~\cite{KK05} that uses the notion of height of an algebraic number, and  which is thus not valid in any positive characteristic. 

In fact, Theorem~\ref{thm:val} holds as soon as $\wronskian(f_1,\dots,f_k)$ does not vanish. 
The difficulty in positive characteristic is that Proposition~\ref{prop:wronskian} does not hold anymore. Yet, the Wronskian is still related to linear independence by the following result (see~\cite{Kaplansky}):
\begin{proposition}\label{prop:positive}
Let $\KK$ be a field of characteristic $p$ and $f_1,\dots,f_k\in\KK[X]$. Then $f_1$, \dots, $f_k$ are linearly independent \emph{over $\KK[X^p]$} if and only if their Wronskian does not vanish.
\end{proposition}

This allows us to give an equivalent of Theorem~\ref{thm:val} in large positive characteristic.

\begin{theorem}\label{thm:valPosChar}
Let $P=\sum_{j=1}^k a_j X^{\alpha_j}(1+X)^{\beta_j} \in\KK[X]$ 
with $\alpha_1\le\dotsb\le\alpha_k$. If 
the characteristic $p$ of $\KK$ satisfies
$p>\max_j(\alpha_j+\beta_j)$, then the valuation of $P$ is at most $\max_j (\alpha_j+\binom{k+1-j}{2})$, 
provided $P$ does not vanish identically.
\end{theorem}

\begin{proof}
Let 
$f_j=X^{\alpha_j}(1+X)^{\beta_j}$ for $1\le j\le k$. 
The proof of Theorem~\ref{thm:val} has two steps: We prove that we can assume that the Wronskian of the $f_j$'s does not vanish, and then under this assumption we get a bound of the valuation of the polynomial. The second part only uses the non-vanishing of the Wronskian and can be used here too. We are left with proving that the Wronskian of the $f_j$'s can be assumed to not vanish when the characteristic is large enough.

Assume that the Wronskian of the  $f_j$'s 
is zero: 
By Proposition~\ref{prop:wronskian}, there is a vanishing linear combination of the $f_j$'s 
with coefficients $b_j$ in
$\KK[X^p]$. Let us write $b_j= \sum b_{i,j}X^{ip}$. Then  $\sum_i X^{ip}\sum_j b_{i,j}f_j=0$.
Since $\deg f_j=\alpha_j+\beta_j<p$, $\sum_j b_{i,j}f_j=0$ for all $i$.
We have thus proved that there is a linear combination of the $f_j$'s equal to zero with coefficients
in $\KK$. Therefore, we can assume we have a basis of the $f_j$'s whose Wronskian
is nonzero and use the same argument as for the  characteristic zero.
\end{proof}

Based on this result, the algorithms we develop in characteristic zero for PIT and factorization can be used for large enough characteristics. 
Computing with lacunary polynomials in positive characteristic has been shown to be hard in many cases~\cite{vzGKS96,KaShp99,KiSha99,KK05,BCR12,KL13}. In particular, it is shown in a very recent paper that it is $\NP$-hard to find roots in $\FFp$ for polynomials over $\FFp$~\cite{BCR12}.

Let $\FFps$ be the field with $p^s$ elements for $p$ a prime number and $s>0$. In the algorithms, it is given as $\FFp[\xi]/\langle\varphi\rangle$ where $\varphi$ is a monic irreducible polynomial of degree $s$ with coefficients in $\FFp$.

\begin{theorem}
Let $P=\sum_{j=1}^k a_j X^{\alpha_j}(uX+v)^{\beta_j} \in\FFps[X]$, 
where $p>\max_j(\alpha_j+\beta_j)$.
There exists a polynomial-time deterministic algorithm to test if $P$ vanishes identically.
\end{theorem}

The proof of this theorem is very similar to the proof of Theorem~\ref{thm:pit}, using Theorem~\ref{thm:valPosChar} instead of Theorem~\ref{thm:val}. 
The main difference 
occurs when $u=0$ or $v=0$. In these cases, we rely in characteristic zero on an external algorithm to test sums of the form $\sum_j a_jv^{\beta_j}$ for zero. This external algorithm does not work in positive characteristic, but these tests are actually much simpler. These sums can be evaluated using repeated squaring in time polynomial in $\log\beta_j$, that is polynomial in the input length.

Note that the condition $p>\max_j(\alpha_j+\beta_j)$ means that $p$ has to be greater than the degree of $P$. This condition is a fairly natural condition for many algorithms dealing with polynomials over finite fields, especially prime fields, for instance for root finding algorithms~\cite{BCR12}.

The basic operations in the algorithm are operations in the ground field $\FFp$. Therefore, the result also holds if bit operations are considered. The only place where computations in $\FFps$ have to be performed in the algorithm is the tests for zero of coefficients of the form $\sum_j\binom{\alpha_j}{\ell_j} a_j u^{-\alpha_j}(-v)^{\ell_j}$ 
where the $\alpha_j$'s and $\ell_j$'s are integers and $a_j\in\FFps$, and the sum has at most $k$ terms. The binomial coefficient is to be computed \emph{modulo} $p$ using for instance Lucas' Theorem~\cite{Lucas1878}.

We now turn to the problem of finding linear factors of lacunary bivariate polynomials.

\begin{theorem}\label{thm:FactorPosChar}
Let $P=\sum_j a_j X^{\alpha_j} Y^{\beta_j}\in\FFps[X,Y]$, where $p>\max_j(\alpha_j+\beta_j)$. 
There exists a probabilistic polynomial-time algorithm to find all the linear factors of $P$ 
of the form $(uX+vY+w)$ with $uvw\neq 0$.

Furthermore, deciding the existence of factors of the form $(X-w)$, $(Y-w)$ or $(X-wY)$ with $w\neq 0$ is $\NP$-hard under randomized reductions.
\end{theorem}

\begin{proof}
The second part of the theorem is the consequence of the $\NP$-hardness (under randomized reductions) of finding roots in $\FFps$ of lacunary univariate polynomials with coefficients in $\FFps$~\cite{KiSha99,BCR12,KL13}: Let $Q$ be a lacunary univariate polynomial over $\FFps$, and define $P(X,Y)=Q(X)$. Then $P$ has the same form as in the theorem with $\beta_j=0$ for all $j$, and factors of the form $(X-w)$ of $P$ are in one-to-one correspondence with roots $w$ of $Q$. Thus, detecting such factors is $\NP$-hard under randomized reductions. The same applies to factors of the form $(Y-w)$. Finally, let us now define $P$ as the homogeneization of $Q$, that is $P(X,Y)=Y^{\deg(Q)}Q(X/Y)$. Then, $P(wY,Y)=Y^{\deg(Q)}P(w,1)=Y^{\deg(Q)}Q(w)$. In other words, factors of $P$ of the form $(X-wY)$ correspond to roots $w$ of $Q$. Thus detecting such factors is also $\NP$-hard under randomized reduction.

For the first part, the algorithm we propose is actually the same as in characteristic zero (Theorem~\ref{thm:factorization}). This means that it relies on known results for factorization of dense polynomials. Yet, the only polynomial-time algorithms known for factorization in positive characteristic are probabilistic~\cite{vzGG03}. Therefore our algorithm is probabilistic and not deterministic as in characteristic zero. 
\end{proof}

\newpage
\appendix

\section{Proof of Theorem~\ref{thm:generalization}} \label{app:generalization}

Let $P_j=\prod_{i=1}^m f_i^{\alpha_{ij}}$ for $1\le j\le k$. As in the proof of Theorem~\ref{thm:val}, we first assume that the $P_j$'s are linearly independent, and the $\alpha_{ij}$'s not less than $(k-1)$. 

We can use a generalized Leibniz rule to compute the derivatives of the $P_j$'s. Namely
\begin{equation}\label{eq:leibnizP}
P_j^{(l)}=\sum_{t_1+\dotsb+t_m=l} \binom{l}{t_1,\dotsc,t_m} \prod_{i=1}^m (f_i^{\alpha_{ij}})^{(t_i)},
\end{equation}
where $\binom{l}{t_1,\dotsc,t_m}$ is the multinomial coefficient. 
Consider now a derivative of the form $(f^\alpha)^{(t)}$. This is a sum of terms, each of which contains a factor $f^{\alpha-t}$. (The worst case happens when $t$ \emph{different} copies of $f$ have been each derived once.) 
In Equation~\eqref{eq:leibnizP}, each $t_i$ is bounded by $l$. This means that $P_j^{(l)}=Q_{l,j} \prod_i f_i^{\alpha_{ij}-l}$ for some polynomial $Q_{l,j}$. 
Since the degree of $P_j^{(l)}$ equals $\sum_i d_i\alpha_{ij}-l$, 
$Q_{l,j}$ 
has degree $\sum_id_i\alpha_{ij}-l-\sum_i(d_i\alpha_{ij}-d_il)=(\sum_id_i-1)l$. 

Consider now the Wronskian $\wronskian$ of the $P_j$'s. We can factor out 
in each column $\prod_i f_i^{\alpha_{ij}-k+1}$ and in each row $\prod_if_i^{k-1-l}$. 
At row $l$ and column $j$, we therefore factor out $\prod_i f_i^{\alpha_{ij}-k+1} \cdot\prod_i f_i^{k-1-l}=\prod_i f_i^{\alpha_{ij}-l}$.
Thus, 
\[\wronskian = \prod_{i=1}^m f_i^{\sum_j\alpha_{ij}-\binom{k}{2}} \det M\]
where $M_{l,j}=Q_{l,j}$. Thus, $\det M$ is a polynomial of degree at most $(\sum_id_i-1)\binom{k}{2}$. 

Therefore, the multiplicity $\mu_\wronskian(\xi)$ of $\xi$ as a root of $\wronskian$ is bounded by its multiplicity as a root of $\prod_i f_i^{\sum_j\alpha_{ij}-\binom{k}{2}}$ plus the degree of $\det M$. We get
\begin{align} 
\mu_\wronskian(\xi)
    &\le \sum_i \mu_i\left(\sum_j\alpha_{ij}-\binom{k}{2}\right)+(\sum_id_i-1)\binom{k}{2} \nonumber\\
    &=\sum_i\left(\mu_i\sum_j\alpha_{ij}+(d_i-\mu_i)\binom{k}{2}\right)-\binom{k}{2}.\label{eq:multUpBd}
\end{align}

To conclude the proof, it remains to remember Lemma~\ref{lemma:valinf} and use the same proof technique as in Theorem~\ref{thm:val}. It was expressed in terms of the valuation of the polynomials, but remains valid with the multiplicity of a root. In this case, it can be written as $\mu_\wronskian(\xi)\ge\sum_j\mu_{P_j}(\xi) -\binom{k}{2}$ where $\wronskian$ is the Wronskian of the $P_j$'s. Using column operations, we can replace the first column of the Wronskian matrix of the $P_j$'s by the polynomial $P$ and its derivatives. We get $\mu_\wronskian(\xi)\ge\mu_P(\xi)+\sum_{j\ge 2}\mu_{P_j}(\xi)-\binom{k}{2}$, where $\mu_{P_j}(\xi)=\sum_i\mu_i\alpha_{ij}$.

Together with \eqref{eq:multUpBd}, 
we get
\begin{align*}
\mu_P(\xi) & \le \mu_\wronskian(\xi)-\sum_{j\ge 2}\mu_{P_j}(\xi)+\binom{k}{2} \\
           & \le \sum_i\left(\mu_i\sum_j\alpha_{ij}+(d_i-\mu_i)\binom{k}{2}\right)-\binom{k}{2}
                                                -\sum_{j\ge 2}\sum_i\mu_i\alpha_{ij} 
                                                +\binom{k}{2}\\
           & \le \sum_i\left(\mu_i\alpha_{i1} +(d_i-\mu_i)\binom{k}{2}\right).
\end{align*}

It remains to remove our two assumptions. If the $P_j$'s are not linearly independent, we can extract a basis $(P_{j_1},\dots,P_{j_d})$. We obtain $\mu_P(\xi)\le\sum_i\left(\mu_i\alpha_{ij_1}+(d_i-\mu_i)\binom{d}{2}\right)$. Since $d\le k+1-j_1$, we have 
\[\mu_P(\xi)\le \max_{1\le j\le k} \sum_{i=1}^m \left(\mu_i\alpha_{ij} + (d_i-\mu_i)\binom{k+1-j}{2}\right).\]
The second assumption is that $\alpha_{ij}\ge k-1$ for all $i$ and $j$. Let 
\[\tilde P = P\cdot\prod_if_i^{k-1}=\sum_j a_j\prod_i f_i^{\tilde{\alpha_{ij}}}.\] 
Since $\tilde{\alpha_{ij}}=\alpha_{ij}+k-1\ge k-1$,
\begin{align*}
\mu_{\tilde P}(\xi)&\le \max_{1\le j\le k} \sum_{i=1}^m \left(\mu_i\tilde{\alpha_{ij}} + (d_i-\mu_i)\binom{k+1-j}{2}\right)\\
    & = (k-1)\sum_{i=1}^m\mu_i + \max_{1\le j\le k} \sum_{i=1}^m \left(\mu_i \alpha_{ij} + (d_i-\mu_i)\binom{k+1-j}{2}\right).
\end{align*}
Since $\mu_{\tilde P}(\xi)=\mu_P(\xi)+(k-1)\sum_i\mu_i$, the result follows.

\begin{remark}
The ordering of the $P_j$'s in the theorem is arbitrary. Yet the value of the bound depends on this ordering. Therefore, it is possible to optimize this bound by using the ordering on the $P_j$'s that minimizes the bound. Let us define
\[s_j=\sum_i\left(\mu_i\alpha_{ij}+(d_i-\mu_i)\binom{k+1-j}{2}\right).\]
The theorem states that $\mu_P(\xi)\le\max_j s_j$. Let $j_1<j_2$ such that $\sum_i \mu_i\alpha_{ij_1}\ge\sum_i\alpha_{ij_2}$. Then $s_{j_1}>s_{j_2}$. These two terms appear in the maximum when $P_{j_1}$ is before $P_{j_2}$ in the ordering. If $P_{j_1}$ and $P_{j_2}$ are exchanged, the two terms are replaced by $\sum_i(\mu_i\alpha_{ij_1}+(d_i-\mu_i)\binom{k+1-j_2}{2})$ and $\sum_i(\mu_i\alpha_{ij_2}+(d_i-\mu_i)\binom{k+1-j_1}{2}$. Neither term is greater than $s_{j_1}$. This means that an exchange of $P_{j_1}$ and $P_{j_2}$ in the ordering cannot increase the bound in the theorem. 

This proves that to minimize the bound the $P_j$'s must be ordered with respect to the value of $\sum_i\mu_i\alpha_{ij}$. This is consistent with the order on the $\alpha_j$'s chosen in Theorem~\ref{thm:val}. We also note that the bound in Theorem~\ref{thm:val} is exactly recovered as a special case. 
\end{remark}


\begin{thebibliography}{10}

\bibitem{BCR12}
J.~Bi, Q.~Cheng, and J.~M. Rojas.
\newblock {Sub-Linear Root Detection, and New Hardness Results, for Sparse
  Polynomials Over Finite Fields}.
\newblock In {\em Proc. ISSAC}, 2013.
\newblock \href{http://arxiv.org/abs/1204.1113}{arXiv:1204.1113}.

\bibitem{BoDu10}
A.~Bostan and P.~Dumas.
\newblock {W}ronskians and linear independence.
\newblock {\em Am. Math. Mon.}, 117(8):722--727, 2010.

\bibitem{CKS99}
F.~Cucker, P.~Koiran, and S.~Smale.
\newblock {A polynomial time algorithm for Diophantine equations in one
  variable}.
\newblock {\em J. Symb. Comput.}, 27(1):21--30, 1999.

\bibitem{FGS08}
M.~Filaseta, A.~Granville, and A.~Schinzel.
\newblock {Irreducibility and Greatest Common Divisor Algorithms for Sparse
  Polynomials}.
\newblock In {\em Number Theory and Polynomials}, volume 352 of {\em P. Lond.
  Math. Soc.}, pages 155--176. Camb. U. Press, 2008.

\bibitem{Gao03}
S.~Gao.
\newblock Factoring multivariate polynomials via partial differential
  equations.
\newblock {\em Math. Comput.}, 72(242):801--822, 2003.

\bibitem{GR08}
M.~Giesbrecht and D.~S. Roche.
\newblock On lacunary polynomial perfect powers.
\newblock In {\em Proc. ISSAC'08}, pages 103--110. ACM, 2008.

\bibitem{GR11}
M.~Giesbrecht and D.~S. Roche.
\newblock Detecting lacunary perfect powers and computing their roots.
\newblock {\em J. Symb. Comput.}, 46(11):1242 -- 1259, 2011.

\bibitem{Hajos}
G.~Hajós.
\newblock [solution to problem 41] (in hungarian).
\newblock {\em Mat. Lapok}, 4:40--41, 1953.

\bibitem{Kal85}
E.~Kaltofen.
\newblock {Polynomial-Time Reductions from Multivariate to Bi- and Univariate
  Integral Polynomial Factorization}.
\newblock {\em SIAM J. Comput.}, 14(2):469--489, 1985.

\bibitem{KK05}
E.~Kaltofen and P.~Koiran.
\newblock On the complexity of factoring bivariate supersparse (lacunary)
  polynomials.
\newblock In {\em Proc. ISSAC'05}, pages 208--215. ACM, 2005.

\bibitem{KK06}
E.~Kaltofen and P.~Koiran.
\newblock Finding small degree factors of multivariate supersparse (lacunary)
  polynomials over algebraic number fields.
\newblock In {\em Proc. ISSAC'06}, pages 162--168. ACM, 2006.

\bibitem{KL13}
E.~L. Kaltofen and G.~Lecerf.
\newblock {Factorization of Multivariate Polynomials}.
\newblock In {\em Handbook of Finite Fields}, Disc. Math. Appl. CRC Press,
  2013.
\newblock To appear.

\bibitem{KN11}
E.~L. Kaltofen and M.~Nehring.
\newblock Supersparse black box rational function interpolation.
\newblock In {\em Proc. ISSAC'11}, pages 177--186. ACM, 2011.

\bibitem{Kaplansky}
I.~Kaplansky.
\newblock {\em An introduction to differential algebra}.
\newblock Actualit{\'e}s scientifiques et industrielles. Hermann, 1976.

\bibitem{KaShp99}
M.~Karpinski and I.~Shparlinski.
\newblock On the computational hardness of testing square-freeness of sparse
  polynomials.
\newblock In {\em Applied Algebra, Algebraic Algorithms and Error-Correcting
  Codes}, volume 1719 of {\em LNCS}, pages 731--731. Springer, 1999.

\bibitem{KS11}
N.~Kayal and C.~Saha.
\newblock {On the Sum of Square Roots of Polynomials and Related Problems}.
\newblock In {\em Proc. CCC'11}, pages 292--299. IEEE, 2011.

\bibitem{KiSha99}
A.~Kipnis and A.~Shamir.
\newblock {Cryptanalysis of the HFE public key cryptosystem by
  relinearization}.
\newblock In {\em Proc. CRYPTO}, pages 19--30. Springer, 1999.

\bibitem{KPT12}
P.~{Koiran}, N.~{Portier}, and S.~{Tavenas}.
\newblock {A Wronskian approach to the real $\tau$-conjecture}.
\newblock \href{http://arxiv.org/abs/1205.1015}{arXiv:1205.1015}, 2012.
\newblock Accepted for oral presentation at {MEGA 2013}.

\bibitem{Lec07}
G.~Lecerf.
\newblock Improved dense multivariate polynomial factorization algorithms.
\newblock {\em J. Symb. Comput.}, 42(4):477--494, 2007.

\bibitem{Len87}
A.~K. Lenstra.
\newblock {Factoring Multivariate Polynomials over Algebraic Number Fields}.
\newblock {\em SIAM J. Comput.}, 16(3):591--598, 1987.

\bibitem{Len99}
H.~Lenstra~Jr.
\newblock Finding small degree factors of lacunary polynomials.
\newblock In {\em Number theory in progress}, pages 267--276. De Gruyter, 1999.

\bibitem{Lucas1878}
{\'E}.~Lucas.
\newblock Théorie des fonctions numériques simplement périodiques.
\newblock {\em Amer. J. Math.}, 1(2--4):184--240,289--321, 1878.

\bibitem{MS77}
H.~Montgomery and A.~Schinzel.
\newblock {Some arithmetic properties of polynomials in several variables}.
\newblock In {\em {Transcendence Theory: Advances and Applications}},
  chapter~13, pages 195--203. Academic Press, 1977.

\bibitem{PetkovsekWilfZeilberger}
M.~Petkov{\v s}ek, H.~S. Wilf, and D.~Zeilberger.
\newblock {\em A=B}.
\newblock AK Peters, 1996.

\bibitem{Pla77}
D.~Plaisted.
\newblock Sparse complex polynomials and polynomial reducibility.
\newblock {\em J. Comput. Syst. Sci.}, 14(2):210--221, 1977.

\bibitem{vzGG03}
J.~von~zur Gathen and J.~Gerhard.
\newblock {\em Modern Computer Algebra}.
\newblock Camb. U. Press, 2nd edition, 2003.

\bibitem{vzGKS96}
J.~von~zur Gathen, M.~Karpinski, and I.~Shparlinski.
\newblock Counting curves and their projections.
\newblock {\em Comput. Complex.}, 6(1):64--99, 1996.

\end{thebibliography}
\end{document}